\documentclass[runningheads,a4paper]{llncs}

\usepackage[english]{babel}
\usepackage[T1]{fontenc}
\usepackage[utf8x]{inputenc}
\usepackage{amsmath,amsfonts}
\usepackage{xspace}
\usepackage[langfont=typewriter]{complexity}
\usepackage{geometry}
\usepackage{url}

\newcommand\CC{\ensuremath{\mathbb C}\xspace}
\newcommand\ZZ{\ensuremath{\mathbb Z}\xspace}
\newcommand\QQ{\ensuremath{\mathbb Q}\xspace}
\newcommand\FF{\ensuremath{\mathbb F}\xspace}
\newcommand\FFp{\ensuremath{\mathbb F_p}\xspace}
\newcommand\KK{\ensuremath{\mathbb K}\xspace}

\newcommand\hn{\ensuremath{\lang{HN}}\xspace}
\newcommand\hhn{\ensuremath{\lang{H}_\lang{2}\lang N}\xspace}
\newcommand\hhnsq{\ensuremath{\lang{H}_\lang{2}\lang N^\square}\xspace}

\title{The Multivariate Resultant is $\NP$-hard\\in any Characteristic}
\author{Bruno Grenet \and Pascal Koiran \and Natacha Portier\thanks{A part of this work was done during visits to the Fields Institute and 
the University of Toronto. It was partially funded by the Fields Institute and the European Community (7th PCRD Contract: PIOF-GA-2009-236197).}}
\institute{%
    LIP\thanks{UMR 5668 \'Ecole Normale Sup\'erieure de Lyon -- CNRS -- UCBL -- INRIA.}, %
    \'Ecole Normale Supérieure de Lyon, Université de Lyon\\%
    and Department of Computer Science, University of Toronto %
    \email{\{Bruno.Grenet,Pascal.Koiran,Natacha.Portier\}@ens-lyon.fr}}
\date{}

\titlerunning{The Multivariate Resultant is $\NP$-hard in any Characteristic}
\authorrunning{B. Grenet, P. Koiran, N. Portier}

\begin{document}

\maketitle

\vspace{3em}

\begin{center}
\framebox[1.1\width]{Rapport de Recherche RRLIP2009-34}
\end{center}

\vspace{2em}

\abstract{The multivariate resultant is a fundamental tool of computational algebraic geometry. It can in particular be used to decide whether a system of $n$ homogeneous equations in $n$ variables is satisfiable (the resultant is a polynomial in the system's coefficients which vanishes if and only if the system is satisfiable). In this paper we present several $\NP$-hardness results for testing whether a multivariate resultant vanishes, or equivalently for deciding whether a square system of homogeneous equations is satisfiable. 
Our main result is that testing the resultant for zero is $\NP$-hard under deterministic reductions in any characteristic, for systems of low-degree polynomials
with coefficients in the ground field (rather than in an extension).
We also observe that in characteristic zero, 
this problem is in the Arthur-Merlin class $\AM$ if the generalized Riemann hypothesis holds true. In positive characteristic, the best upper bound remains $\PSPACE$.}

\section{Introduction} 

Given two univariate polynomials, their Sylvester matrix is a matrix built on the coefficients of the polynomials which is singular iff the polynomials
have a common root. The determinant of the Sylvester matrix is known as the resultant of the polynomials. 
This determinant is easy to compute since the size of the Sylvester matrix 
is the sum of the degrees of the polynomials.
The study of the possible generalizations to multivariate systems comes within the scope
of the theory of elimination \cite{vanderWaerden,Mac02,Dix08,Macaulay,Stu91,EM99}. 
This theory proves that the only case where a unique polynomial can testify to the existence
of a common root to the system is the case of $n$ homogeneous polynomials in $n$ variables: 
the resultant of a square system of homogeneous polynomials $f_1,\dots,f_n\in\KK[x_1,\dots,x_n]$ is a polynomial in the indeterminate coefficients of 
$f_1,\dots,f_n$ which vanishes iff $f_1,\dots,f_n$ have a nonzero common root
in the algebraic closure of $\KK$. 
The resultant of such a system is known as the \emph{multivariate resultant} in the literature. 
This captures the case of two univariate polynomials \emph{modulo} their homogenization.
Furthermore, in many cases a system of more than $n$ homogeneous polynomials in $n$ variables can be
reduced to a system of $n$ homogeneous polynomials, so that the square case is an important one. This result is sometimes known as Bertini's theorem
(as explained toward the end of this section, we will use an effective version
of this result in one of our $\NP$-hardness proofs).
In this paper, we focus on the multivariate resultant which we simply refer to as the resultant.

The resultant has been extensively used to solve polynomial systems \cite{Laz81,Ren89,CKY89,CD05} and for the elimination of quantifiers in algebraically
or real-closed fields \cite{Sei54,Ier89}.
More recently, the multivariate resultant has been of interest in pure and applied domains. For instance, the problem of robot motion planning is
closely related to the multivariate resultant \cite{Can87,Can88-robot,CR87}, and more generally the multivariate resultant is used in real algebraic 
geometry \cite{Can88-pspace,KSY94}. 
Finally, in the domain of symbolic computation progress has been made for finding explicit formulations for the resultant 
\cite{KS95,Can89,DD01,BD04,CD05,JS07}, see also \cite{KK08}. 

\begin{definition}
Let \KK be a field and $f_1,\dots,f_n$ be $n$ homogeneous polynomials in $\KK[x_1,\dots,x_n]$, $f_i(\bar x)=\sum_{|\alpha|=d_i}\gamma_{i,\alpha}x^\alpha$. 
The \emph{multivariate resultant} $R$ of $f_1,\dots,f_n$ is an irreducible polynomial in $\KK[\overline{\gamma_{i,\alpha}}]$ such that 
\begin{equation} R(\overline{\gamma_{i,\alpha}})=0 \iff \exists\bar x\in\bar\KK, f_1(\bar x)=\cdots=f_n(\bar x)=0. \end{equation}
The multivariate resultant is unique up to a constant factor.
\end{definition}

The problem we are interested in is testing the resultant for zero. This is the same as deciding whether a square system of homogeneous polynomials 
(that is $n$ polynomials in $n$ variables) has a non-trivial root. This is closely related to 
the decision problem problem for the existential theory of an algebraically closed field. This problem is sometimes called the 
\emph{Hilbert Nullstellensatz} problem:

\begin{definition}
Let \KK be a field and $\bar\KK$ be an algebraic closure of \KK. The \emph{Hilbert Nullstellensatz} problem over \KK, $\hn(\KK)$, is the following:
Given a system $f$ of $s$ polynomials in $\KK[x_0,\dots,x_n]$, does there exist a root of $f$ in $\bar\KK^{n+1}$?

Let us now assume  that the $s$ components of $f$ are homogeneous polynomials. Then the \emph{homogeneous} Hilbert Nullstellensatz problem over \KK, $\hhn(\KK)$, is to
decide whether a \emph{non trivial} (that is, nonzero) root exists in $\bar\KK$.

If $f$ is supposed to contain as many homogeneous polynomials as variables, the problem is called the \emph{square} homogeneous Hilbert Nullstellensatz 
over \KK, $\hhnsq(\KK)$.
\end{definition}

In the case of the field $\mathbb Q$, it is more natural to have coefficients in \ZZ. We shall use the notations \hn, \hhn and \hhnsq for this case where
the system is made of integer polynomials. In the sequel, for any prime number $p$, the finite field with $p$ elements is denoted by $\FF_p$. The notation
is extended to characteristic zero, and $\FF_0=\mathbb Q$.

In the case of polynomials with coefficients in \ZZ, Canny \cite{Can88-robot} gave in 1987 a $\PSPACE$ algorithm to compute the resultant. 
To the authors' knowledge, this is the best known upper bound. 
In this paper we show that testing the resultant for zero is $\NP$-hard in
any characteristic.
In other words, $\hhn(\KK)$ is $\NP$-hard for any field $\KK$.

\subsection*{Main Results and Proof Techniques}

In Section~\ref{char0} we observe that for polynomials with integer coefficients, testing the resultant for zero is a problem in the Arthur-Merlin 
($\AM$) class. This result assumes the generalized Riemann hypothesis, and follows
from a simple reduction to the Hilbert Nullstellensatz. For this problem, membership in 
$\AM$ assuming GRH was established in~\cite{Koi96}.
The remainder of the paper is devoted to hardness results.

In characteristic zero, it seems to be a ``folklore'' result that testing the resultant for zero is $\NP$-hard. 
We give a proof of this fact in Section~\ref{char0} since we have not been able to find one in the literature.
In fact, we give two proofs of two results of incomparable strength.
The first proof is based on a reduction from the  \lang{Partition}
problem  \cite[problem SP12]{GareyJohnson}. The second proof is based on a result of
Plaisted~\cite{Pla84} and shows that the problem remains $\NP$-hard for systems
of only two homogeneous polynomials. For the latter result to be true,
we need to use a sparse encoding for our two polynomials (their degree
can therefore be exponential in the input size).

The first proof does not carry over to positive characteristic since
the $\NP$-hardness of \lang{Partition} relies in an essential way on the fact
that the data are integers (in fact, in any finite field the analogue
problem can be solved in polynomial time by dynamic programming).

Plaisted's result can be adapted to positive characteristic~\cite{GKS96,KK05}
but this requires randomization. By contrast, our ultimate goal is $\NP$-hardness
for deterministic reductions and low degree polynomials.
We therefore need to use different techniques.
Our starting point is a fairly standard encoding of $\lang{3-SAT}$
by systems of polynomial equations. Using this encoding we show at the beginning of Section~\ref{arbitrary} that deciding the existence of a nontrivial
solution to a system of homogeneous equations is $\NP$-hard in any characteristic.
The resulting system has in general more equations than variables.
In order to obtain a square system 
two basic strategies can be explored:
\begin{itemize}
\item[(i)] Decrease the number of equations.
\item[(ii)] Increase the number of variables.
\end{itemize}
In Section~\ref{sec:rand} we give a randomized 
$\NP$-hardness result based on the first
strategy. The idea is to replace the initial system by
a random linear combinations of the system's equations 
(the fact this does not change
the solution set is sometimes called a ``Bertini's theorem'').

In Section~\ref{arbitrary} we use the second strategy to 
obtain two $\NP$-hardness results for deterministic reductions.
The main difficulty is to make sure that the introduction of new variables
does not create spurious solutions 
(we do not want to turn an unsatisfiable system into a satisfiable system).
Our solution to this problem can be viewed as a derandomization result.
Indeed, it can be shown that the coefficients of the monomials where 
the new variables occur could be chosen at random.
It would be interesting to find out whether the proof based on the first
strategy can also be derandomized.





\section{Complexity of the Resultant in Characteristic~$0$}
\label{char0}

In this section we show that  testing the resultant for zero is reducible to $\hn(\KK)$. In the case 
$\KK = \ZZ$, this allows us
to conclude (under the Generalized Riemann Hypothesis) 
that our problem is in the polynomial hierarchy, and more precisely in 
the Arthur-Merlin class.
In fact, we show that this applies more generally to the satisfiability problem
for homogeneous systems (recall that testing the resultant for zero 
corresponds to the square case).

\begin{proposition}\label{prop:in-am}
For any field \KK, the problem $\hhn(\KK)$ is polynomial-time many-one reducible to $\hn(\KK)$.
\end{proposition}

\begin{proof}
Consider an instance $\mathcal S$ of $\hhn(\KK)$, that is $s$ homogeneous polynomials 
$f_1,\dots,f_s\in\KK[x_1,\dots,x_n]$. The polynomials $f_1,\dots,f_s$ can be viewed as elements of $\KK[x_1,\dots,x_n,y_1,\dots,y_n]$ 
where $y_1,\dots,y_n$ are new variables which do not appear in the $f_i$. Let $\mathcal T$ be the system containing all the
$f_i$ and the new (non-homogeneous) polynomial $\sum_{i=1}^n x_iy_i-1$. This is an instance of the problem $\hn(\KK)$. It remains to prove 
that $\mathcal S$ and $\mathcal T$ are equivalent.

Given a root $(a_1,\dots,a_n,b_1,\dots,b_n)$ of $\mathcal T$, the new polynomial ensures that there is at least one nonzero $a_i$. 
So $(a_1,\dots,a_n)$ is a non trivial root of 
$\mathcal S$. Conversely, suppose that $\mathcal S$ has a non trivial root $(a_1,\dots,a_n)$, and let $i$ be such that $a_i\ne 0$. Then the tuple
$(a_1,\dots,a_n,0,\dots,0,a_i^{-1},0,\dots,0)$ where $a_i^{-1}$ corresponds to the variable $y_i$ is a root of $\mathcal T$.

Thus $\hhn(\KK)$ is polynomial-time many-one reducible to $\hn(\KK)$.
\end{proof}

Koiran \cite{Koi96} proved that $\hn\in\AM$ under the Generalized Riemann Hypothesis.
We denote here by $\AM$ the Arthur-Merlin class, defined by \emph{interactive proofs with public coins} (see \cite{AroraBarak}).
Thereby,

\begin{corollary}\label{corr:in-am}
Under the Generalized Riemann Hypothesis, \hhn is in the class $\AM$.
\end{corollary}

In positive characteristic, the best upper bound on the complexity of the Hilbert Nullstellensatz known to this day remains  $\PSPACE$ (in particular
it is not known whether the problem lies in the polynomial hierarchy,
even assuming some plausible number-theoretic conjecture such as the generalized Riemann hypothesis).



We now give our first $\NP$-hardness result, for the satisfiability of square
systems  of homogeneous polynomial equations. As explained in the introduction,
this seems to be a ``folklore'' result. We give the (short) proof since 
 finding an explicit statement (and proof) of this result in the literature
appears to be difficult.
The second part of the theorem shows that the problems remains $\NP$-hard even
for systems with small integer coefficients (i.e., coefficients bounded by $2$).
This is achieved by a standard trick: we introduce new variables in order to ``simulate'' large integers coefficients. 
It is interesting to note, however, that
a similar trick for reducing degrees does not seem to apply 
to the resultant problem (more on this after Theorem~\ref{plaisted}).
\begin{theorem}\label{thm:folklore}

The problem \hhnsq of deciding whether a square system of homogeneous polynomials with coefficients in \ZZ has a non trivial root is $\NP$-hard.

The problem remains $\NP$-hard even if no polynomial has degree greater that $2$ and even if the coefficients are bounded by $2$.
\end{theorem}

\begin{proof} 
The reduction is done from \lang{Partition} which is known to be $\NP$-hard \cite[problem SP12]{GareyJohnson}: Given a finite set $S$ and a non negative 
integer
weight $w(s)$ for each $s\in S$, the problem is to decide the existence of subset $S'$ such that $\sum_{s\in S'} w(s)=\sum_{s\notin S'} w(s)$. That is, the aim 
is to cut $A$ into two subsets of same weights.

Given such an instance of \lang{Partition} where $S=\{s_1,\dots,s_n\}$, let us define a system of polynomials. For $1\le i\le n$, $f_i(\bar x)=x_0^2-x_i^2$. 
And \begin{equation} f_0(x_0,\dots,x_n)=w(s_1)x_1+w(s_2)x_2+\cdots+w(s_n)x_n. \end{equation}
A tuple $(a_0,\dots,a_n)$ has to verify $a_i=\pm a_0$ for each $i$ to be a solution, hence the only case to consider is $a_0=1$ and $a_i=\pm1$ for $i\ge 1$.
Then it is clear that the system has a non trivial solution iff $S$ may be split into two subsets of equal weights.

For the second part of the theorem, it remains to show that the coefficients in the system can be bounded by $2$. As the $w(s_i)$ may be big integers, 
they have to be replaced by variables. Let us write $w(s_i)=\sum_{j=0}^p w_{ij}2^j$. 
For each $w_{ij}$, a new variable $W_{ij}$ is introduced. 
For every $i$, the values
of the $W_{ij}$ are defined by a descending recurrence:
\begin{equation} \left\{\begin{array}{l@{\ }l@{\ }lll}
W_{ip}&-&w_{ip}x_0&=&0\\
W_{ij}&-&(2W_{i,j+1}+w_{ij}x_0)&=&0\quad\text{for all $j<p$}
\end{array}\right. \end{equation}
These equalities imply that for every $i,j$ we have $W_{ij}=\sum_{l=j}^p w_{ij} 2^{j-l} x_0$.
Then $f_0$ is replaced by $W_{1,0}x_1+W_{2,0}x_2+\cdots+W_{n,0}x_n$. 
Doing so, the number of polynomials remains the same as the number of variables. 
Hence, this algorithm build a new homogeneous system where the polynomials have their coefficients bounded 
by $2$ and their degrees too. One can readily check that the new system has a non trivial solution iff the original one has. In particular, 
if $x_0$ is set to zero, then all other variables have to be zero too.
\qed\end{proof} 

A related result is Plaisted's \cite{Pla84} on the $\NP$-hardness of deciding whether the gcd of two sparse univariate polynomials has degree 
greater than one. By homogenization of the polynomials, this is the same problem as in Theorem \ref{thm:folklore} for only two bivariate polynomials. 
Note that the polynomials are sparse and can be of very high degree since
exponents are written in binary (this polynomial representation is sometimes called ``supersparse''~\cite{KK05}).
If both polynomials were dense, 
the resultant could be computed in polynomial time
since it is equal to the determinant of their Sylvester matrix.
Plaisted's theorem stated in the same language as Theorem
\ref{thm:folklore} is the following:

\begin{theorem} \label{plaisted}
Given two sparse homogeneous polynomials in $\ZZ[x,y]$, it is $\NP$-hard to decide whether they share a common root in $\CC^2$.
\end{theorem}
We briefly sketch Plaisted's reduction since it will help understand
the discussion at the end of this section. For a full proof (including 
a correctness proof), see \cite[Theorem 5.1]{Pla84}.
\begin{proof}[sketch] 
The idea is to turn a \lang{3-SAT} instance into a system of two univariate polynomials which share a common root iff the 3-CNF formula 
is satisfiable.

To every variable $X_j$ is associated a prime $p_j$, and $M=\prod p_j$ is defined where the product ranges over all variables that appear in the
formula. A formula $\phi$ is turned into a polynomial $P_\phi$ according to the following rules. A non negated variable $X_j$ is turned into 
$P_{X_j}(x)=x^{M/p_j}-1$ and a negated variable $\neg X_k$ into $P_{\neg X_k}(x)=1+x^{M/p_k}+\cdots+x^{(p_k-1)M/p_k}$.  
Then a formula $\phi\vee\psi$ is turned into $P_{\phi\vee\psi}=\text{lcm}(P_\phi,P_\psi)$.
A conjunction  $\phi= \bigwedge_i \phi_i$ is turned into the polynomial 
\begin{equation}\label{eq:plaisted}
P_{\phi}(x)=x^M \sum_i P_{\phi_i}(x)P_{\phi_i}(1/x)
\end{equation}
This defines the first polynomial $P$.
The second polynomial is simply $x^M-1$. The proof that those two polynomials share a common root iff $\phi$ is satisfiable is omitted.

To obtain the result in the way we stated it, it is sufficient to homogenize $P(x)$ and $x^M-1$ with the second variable $y$.
\qed\end{proof} 

Theorems~\ref{thm:folklore} and~\ref{plaisted} seem to be incomparable.
In particular, it is not clear how to derive Theorem~\ref{thm:folklore}
from Theorem~\ref{plaisted}. A natural idea would be to introduce new variables
and use the repeated squaring trick to reduce the degrees of the polynomials
occurring in Plaisted's result. However, as we now explain this can lead
to the creation of unwanted roots at infinity in the resulting 
polynomial system.

Assume for instance that we wish to get rid of all occurrences of $x^2$ in a polynomial. One can add a new variable
$x_2$, replace the occurrences of $x^2$ by $x_2$ and add a new polynomial $x_2-x^2$. In order to keep the system homogeneous, the idea is to homogenize
the latter polynomial: $x_0x_2-x^2$. The problem with this technique is that it adds some new roots with all variables but $x_2$ set to $0$, and in 
particular the homogenization variable $x_0$.

To give an explicit example of the problem mentioned above, let us consider the formula
\begin{equation} (X\vee Y)\wedge(\neg X)\wedge(\neg Y). \end{equation}
Let us associate the prime number $2$ to the variable $X$, and $3$ to $Y$ ($M$ in the previous proof is therefore $6$). 
By Plaisted's construction, $X$ is turned into $x^{M/2}-1=x^3-1$ and $Y$ into $x^2-1$. Their negations $\neg X$ and $\neg Y$ are respectively turned
into $1+x^3$ and $1+x^2+x^4$. The disjunction of $X$ and $Y$ is turned into the lcm of $x^3-1$ and $x^2-1$, that is $(x^2-1)(x^2+x+1)$. Finally, we have
to apply formula \eqref{eq:plaisted} with the latter polynomial, $1+x^3$ and $1+x^2+x^4$. Therefore, the two polynomials of Plaisted's construction
are $x^M-1=x^6-1$ and $-x^3+x^4+2x^5+9x^6+2x^7+x^8-x^9$. 
It can be checked that
as expected,
 those two polynomials do not share any common root.

Applying the repeated squaring trick with homogenization on this example
gives the following system where the two first polynomials represent the original ones and the other ones
are new ones:
\begin{equation} \left\{ \begin{array}{rrr}
\multicolumn{3}{r}{-x_3+x_4+2x_5+9x_6+2x_7+x_8-x_9=0}\\
x_6-x_0=0;& x_0x_2-x^2=0; & x_0x_3-x_2x=0\\
x_0x_4-x_2^2=0;& x_0x_5-x_4x=0; & x_0x_6-x_2x_4=0\\
x_0x_7-x_4x_3=0; & x_0x_8-x_4^2=0; & x_0x_9-x_8x=0
\end{array}\right. \end{equation}
But in that example, one can easily check that solutions with $x_0=0$ exist. Namely if we set $x_8$ and $x_9$ to the same nonzero value and all other 
variables to $0$, this defines a solution to the system. 

To the authors' knowledge, there is no solution to avoid these unwanted roots. Furthermore, Plaisted's result works well with fields of characteristic
$0$, but as it uses the fact that a sum of non negative terms is zero iff every term is zero, this generalizes not so well to positive characteristic.
In particular, generalizations to positive characteristic require randomization (see \cite{KK05} and \cite{GKS96}).
By contrast, two of the reductions given in the next section are deterministic
and they yield systems with polynomials of low degree (i.e., of linear or even constant degree).


\section{The resultant is $\NP$-hard in arbitrary characteristic} 
\label{arbitrary}

In this section we give three increasingly stronger $\NP$-hardness results
for testing the resultant.
As explained in the introduction, we first provide in Section~\ref{sec:rand} a $\NP$-hardness proof for  randomized reductions. We then give in Section~\ref{sec:deter} two $\NP$-hardness 
results for deterministic reductions: the first
one applies to systems with coefficients in an extension of the ground field,
and the second (stronger) result to systems with coefficients in the ground field only.
The starting point for these three $\NP$-hardness results is the following lemma.
\begin{lemma}[\cite{Koi00-dimensions}] \label{lemma:h2n}
Given a field \KK of any characteristic, it is $\NP$-hard to decide whether a system of $s$ homogeneous polynomials in $\KK[x_0,\dots,x_n]$ has a 
non trivial root. That is, $\hhn(\KK)$ is $\NP$-hard.
\end{lemma}
In~\cite{Koi00-dimensions},  $\hhn(\KK)$ was proven $\NP$-hard by reduction from
\lang{Boolsys}. An input of \lang{Boolsys} 
is a system of boolean equations in the variables $X_1,\dots,X_n$ where each equation is of the form 
$X_i=\text{True}$, $X_i=\neg X_j$, or $X_i=X_j\vee X_k$. 
The question is the existence of a valid assignment for the system, that is an assignment
of the variables such that each equation is satisfied. This problem is easily shown $\NP$-hard by reduction from $\lang{3-SAT}$.
We now give a proof of this lemma since the specific form of the systems 
that we construct in the reduction will be useful in the sequel.
This proof is a slight variation on the proof from~\cite{Koi00-dimensions}.
\begin{proof} 
Let \KK be a field of any characteristic $p$, $p$ being either zero or a prime number. 
At first, $p$ is supposed to be different from $2$. The proof has to be
slightly changed in the case $p=2$ and this case is explained at the end of the proof.

Let $\mathcal B$ be an instance of \lang{Boolsys}. Let us define a system of homogeneous polynomials from this instance with the property that 
$\mathcal B$ is satisfiable iff the polynomial system has a non trivial common root. The variables in the system are $x_0,\dots,x_n$ where $x_i$, 
$1\le i\le n$, corresponds to the boolean variable $X_i$ in \lang{Boolsys}, and $x_0$ is a new variable. The system contains four kinds of polynomials:
\begin{itemize}
\item $x_0^2-x_i^2$, for each $i>0$;
\item $x_0\cdot(x_i+x_0)$, for each equation $X_i=\text{True}$ in \lang{Boolsys};
\item $x_0\cdot(x_i+x_j)$, for each equation $X_i=\neg X_j$;
\item $(x_i+x_0)^2-(x_j+x_0)\cdot(x_k+x_0)$, for each equation $X_i=X_j\vee X_k$.
\end{itemize}

Let us denote by $f$ the polynomial system obtained from $\mathcal B$. The first kind of polynomials
ensures that if $(a_0,\dots,a_n)$ is a non trivial root of $f$, then $a_0^2=a_1^2=\cdots=a_n^2$. Now if $f$ has a non trivial root $(a_0,\dots,a_n)$, 
then one can readily check that the assignment $X_i=\text{True}$ if $a_i=-a_0$ and $X_i=\text{false}$ if $a_i=a_0$ satisfies $\mathcal B$. 
Conversely, if there is a valid assignment $X_1,\dots,X_n$ for $\mathcal B$, any $(n+1)$-tuple $(a_0,\dots,a_n)$ where $a_0\neq 0$ and $a_i=-a_0$ 
if $X_i=\text{True}$ and $a_i=a_0$ if $X_i=\text{false}$ is a non trivial root of $f$.

This proof works for any field of characteristic different from $2$. The problem in characteristic $2$ is the implementation of \lang{Boolsys} in terms of
a system of polynomials. Indeed, for the other characteristics, the truth is represented by $-a_0$ and the falseness by $a_0$. In characteristic $2$, those 
values are equal. Yet, one can just change the polynomials and define in the case of characteristic $2$ the following system:
\begin{itemize}
\item $x_0x_i-x_i^2$, for each $i>0$;
\item $x_0(x_i+x_0)$, for each equation $X_i=\text{True}$ in \lang{Boolsys};
\item $x_0(x_i+x_j+x_0)$, for each equation $X_i=\neg X_j$;
\item $x_i^2+x_jx_k+x_0\cdot(x_j+x_k)$, for each equation $X_i=X_j\vee X_k$.
\end{itemize}
Now, given any nonzero value $a_0$ for $x_0$, the truth of a variable $X_i$ is represented by $x_i=a_0$ whence the falseness is represented by $x_i=0$.
A root of the system is in particular a root of the polynomials defined by the first item. Therefore each $x_i$ has to be set either to $a_0$ or to $0$. 
The system has a non trivial root iff the instance of \lang{Boolsys} is satisfiable. 
\qed\end{proof} 

\subsection{A randomized reduction} \label{sec:rand} 

We now give the first of our three $\NP$-completeness results in positive characteristic. The proof also applies to  characteristic zero, but in this case
Theorem~\ref{thm:folklore} is preferable (its proof is simpler and the $\NP$-hardness result stronger since it relies on deterministic reductions).
For more on randomized reductions, see~\cite{AroraBarak}.
\begin{theorem}\label{thm:randomized}
Let $p$ be either zero or a prime number. 
The following problem is $\NP$-hard under randomized reductions:
\begin{itemize}
\item[-] INPUT: a square system of homogeneous equations with coefficients
in a finite extension of  $\FF_p$.
\item[-] QUESTION: is the system satisfiable in the algebraic closure of $\FF_p$?
\end{itemize}
In
the case $p=0$, the results also holds for systems with coefficients in~$\ZZ$.
\end{theorem}

\begin{proof} 
Lemma \ref{lemma:h2n} shows that it is $\NP$-hard to decide whether a non square polynomial system $f$ with coefficients in \FFp has a non trivial root.
From $f$, a square system $g$ is built in randomized polynomial time.

Let us denote by $f_j$, $1\le j\le s$, the components of $f$. They are homogeneous polynomials in $\FFp[x_0,\dots,x_n]$. The components of $g$ are 
defined by 
\begin{equation} g_i=\sum_{j=1}^s \alpha_{ij}f_j \end{equation}
for $0\le i\le n$. In the sequel, we explain how to choose the $\alpha_{ij}$ for $f$ and $g$ to be equivalent. 
For any choice of the $\alpha_{ij}$, a root of $f$ is a root of $g$. Thus it is sufficient to show how to choose the $\alpha_{ij}$ so that $g$ has no non
trivial root if the same is true for $f$.

The property the $\alpha_{ij}$ have to verify is expressed by a first-order formula:
\begin{equation} \Phi(\bar\alpha)\equiv\forall x_0\cdots\forall x_n 
    \left(\bigwedge_{j=1}^s f_j(\bar x)=0\right)\vee\left(\bigvee_{i=0}^n\sum_{j=1}^s\alpha_{ij}f_j(\bar x)\neq0\right). \end{equation}
The formula $\Phi$ belongs to the language of the first-order theory of 
the algebraically closed field $\overline{\FFp}$. This theory eliminates quantifiers and 
$\Phi(\bar\alpha)$ is therefore equivalent to a quantifier-free formula of the form
\begin{equation} \Psi(\bar\alpha)\equiv\bigvee_k\left(\bigwedge_{l} P_{kl}(\bar\alpha)=0\wedge\bigwedge_m Q_{km}(\bar\alpha)\neq0\right), \end{equation}
where $P_{kl},Q_{km}\in
\FFp[\bar\alpha]$. As a special case of \cite[Theorem 2]{FGM90}, one can bound the number of polynomials in $\Psi$ as well 
as their degrees by $2^{P(n,\log(s+n))}$ where 
$P$ is a polynomial independent from $\Phi$.

The proof of Theorem \ref{thm:variety} (in Appendix)
shows that   the set $A$ of tuples satisfying $\Phi$ is Zariski-dense in $\overline{\FFp}^{s(n+1)}$.
Since $A$ is dense, and $A$ is also defined by $\Psi$, one of the clauses 
of $\Psi$ must define a Zariski dense subset of $\overline{\FFp}^{s(n+1)}$.
This clause is of the form $\bigwedge_m Q_m(\bar\alpha)\neq 0$.

To satisfy $\Phi$, it is sufficient for the $\alpha_{ij}$ to avoid the roots of a polynomial $Q=\prod_m Q_m$. As mentioned before, it is known that
$\Psi$ contains at most $2^{P(n,\log(s+n))}$ polynomials of degree at most $2^{P(n,\log(s+n))}$. Thus, $Q$ is a polynomial of degree at most 
$2^{2P(n,\log(s+n))}$. Consider now a finite extension \KK of \FFp with at least $2^{2+2P(n,\log(s+n))}$ elements (that is, of polynomial degree). 
If we choose the $\alpha_{ij}$ uniformly at random in \KK, then with probability at least $3/4$ they are not a root of $Q$ (by the Schwartz-Zippel Lemma).
Thus with the same probability, they satisfy $\Phi$. Note that $\KK$ can be built in polynomial-time with Shoup's algorithm \cite{Sho90} 
when $p$ is prime (for $p=0$, we take of course $\KK = \QQ$).

To sum up, we build from $f$ a square system $g$ defined by random linear combinations of the components of $f$. If $f$ has a non trivial root, then it is
a root of $g$ too. Conversely, if $f$ has no non trivial root, then with probability at least $3/4$ it is also the case that $g$ has no nontrivial root.
\qed\end{proof} 

In characteristic zero 
the bounds in the above proof can be sharpened: instead of appealing to the general-purpose quantifier elimination result of~\cite{FGM90} we can use 
a result of \cite{KPS01}. Indeed, it follows from section~4.1 of~\cite{KPS01} 
that there exists
a polynomial $F$ of degree at most $3^{n+1}$ such that $F(\bar\alpha)\neq 0$ implies that $g$ has no non trivial root as soon as it is true for $f$.
This polynomial plays the same role as $Q$ in the previous proof but the bound on its degree is sharper.


\subsection{
Deterministic Reductions} \label{sec:deter} 

We now improve the $\NP$-hardness result of Sect.~\ref{sec:rand}: we show
that the same problem is $\NP$-hard {\em for deterministic reductions}.
This result is not only stronger, but also the proof is more elementary
(there is no appeal to effective quantifier elimination).
\begin{theorem}
Let $p$ be either zero or a prime number.
The following problem is $\NP$-hard under deterministic reductions:
\begin{itemize}
\item[-] INPUT: a square system of homogeneous equations with coefficients
in a finite extension of  $\FF_p$.
\item[-] QUESTION: is the system satisfiable in the algebraic closure of $\FF_p$?
\end{itemize}
In the case $p=0$, the results also holds for systems with coefficients in~$\ZZ$.
\end{theorem}

\begin{proof} 
The proof of Lemma \ref{lemma:h2n} gives a method to implement an
instance of \lang{Boolsys} with a system $f$ of $s$ homogeneous polynomials in $n+1$ variables with coefficients in \FFp. 
It remains to explain how to construct a square system $g$ that has a non trivial root iff $f$ does. 
%
Let us denote by $f_1,\dots,f_s$ the components of $f$, 
with for each $i=1,\dots,n$, $f_i=x_0^2-x_i^2$ if $p\neq 2$ and $f_i=x_0x_i- x_i^2$ if $p=2$. A new system $g$ of $s$ polynomials in $s$ variables
is built. The $s$ variables are $x_0,\dots,x_n$ and $y_1,\dots,y_{s-n-1}$, that is $(s-n-1)$ new variables are added. The system $g$ is the following:
\begin{equation} g(\bar x,\bar y)=\left(\begin{array}{l@{}c@{}l}
f_1(\bar x)&&\\
\quad\vdots&&\\
f_n(\bar x)&&\\
f_{n+1}(\bar x)&&+\lambda y^2_1\\
f_{n+2}(\bar x)&-y^2_1&+\lambda y^2_2\\
&\vdots&\\
f_{n+i}(\bar x)&-y_{i-1}^2&+\lambda y_i^2\\
&\vdots&\\
f_{s-1}(\bar x)&-y^2_{s-n-2}&+\lambda y^2_{s-n-1}\\
f_s(\bar x)&-y^2_{s-n-1}&
\end{array}\right) \end{equation}
The parameter $\lambda$ is to be defined later. Clearly, if $f$ has a non trivial root $\bar a$, then $(\bar a,\bar 0)$ is a non trivial root of $g$. 
Let us now prove that the converse also holds true for some $\lambda$: if $g$ has a non trivial root, then so does $f$. Note that a suitable $\lambda$
has to be found in polynomial time.

Let $(a_0,\dots,a_n,b_1,\dots,b_{s-n-1})$ be any non trivial root of $g$. 
Since $\bar a$ must be a common root of $f_1,\dots,f_n$, we
have $a_0^2=\dots=a_n^2$ if $p\neq 2$, and $a_i \in \{0,a_0\}$ for every $i$ if $p=2$. 
  Now, either $a_0=0$ and $f_i(\bar a)=0$ for every $i$, or $a_0$ can be supposed to equal $1$.
Therefore, if $p\neq 2$ either $\bar a=\bar 0$ or $a_i=\pm 1$ for every $i$, and if $p=2$ either $\bar a=\bar 0$ or $a_i\in \{0,1\}$ for every $i$. 
Let us define $\epsilon_i=f_{n+i}(\bar a)\in\FFp$. As $(\bar a,\bar b)$ is a root of $g$, the $b_i^2$ satisfy the linear system
\begin{equation} \left\{\begin{matrix}
\epsilon_1 && &+&  \lambda Y_1 &=&0,\\
\epsilon_2 &-& Y_1 &+& \lambda Y_2 &=&0,\\
&&&\vdots\\
\epsilon_{s-n-1} &-& Y_{s-n-2} &+& \lambda Y_{s-n-1} &=&0,\\
\epsilon_{s-n} &-& Y_{s-n-1} && &=&0.
\end{matrix}\right. \end{equation}
This system can be homogenized by replacing each $\epsilon_i$ by $\epsilon_i Y_0$ where $Y_0$ is a fresh variable. This
gives a square homogeneous linear system. The determinant of the matrix of this system is equal to $(-1)^{s-n-1}\left(\epsilon_1+\epsilon_2\lambda+
\cdots+\epsilon_{s-n}\lambda^{s-n-1}\right)$.

Let us consider this determinant as a polynomial in $\lambda$. This polynomial vanishes identically iff all the $\epsilon_i$ are zero. In that case, the 
only solutions satisfy $Y_i=0$ for $i>0$, that is $(\bar a,\bar 0)$ is a root of $g$ and therefore $\bar a$ is a root of $f$. 
If some $\epsilon_i$ are nonzero, this is a nonzero polynomial of degree $(s-n-1)$. If $\lambda$ can be chosen such that it is not a root of this polynomial
(for any possible nonzero value of $\bar\epsilon$), then the only solution to the linear system is the trivial one. This means that the only non trivial
root of $g$ is $(\bar a,\bar 0)$ where $\bar a$ is a root of $f$.

If the polynomials have coefficients in \ZZ,
$\lambda=3$ (or any other integer $\lambda>2$) satisfies the condition. Indeed, one can check that $\epsilon_i=
f_{n+i}(\bar a)\in\{-4,0,2,4\}$ when $a_0=1$. The determinant is zero iff $\epsilon'_1+\epsilon'_2\lambda+\cdots+\epsilon'_{s-n}\lambda^{s-n-1}=0$
where $\epsilon'_i=\epsilon_i/2\in\{-2,0,1,2\}$.
For each $i$, let $\epsilon_i^+=\max\{\epsilon'_i,0\}$ and $\epsilon_i^-=\max\{-\epsilon'_i,0\}$. Then $\epsilon'_i=\epsilon_i^+-\epsilon_i^-$, 
and $0\le \epsilon_i^+,\epsilon_i^-\le 2$. Now the determinant is zero iff $\sum_i\epsilon_i^+ 3^i=\sum_i\epsilon_i^- 3^i$. By the
unicity of base-$3$ representation, this means that for all $i$, $\epsilon_i^+=\epsilon_i^-$, and so $\epsilon'_i=0$.

For a field of positive characteristic, this argument cannot be applied. The idea is to find a $\lambda$ that is not a root of any polynomial
of degree $(s-n-1)$. Nothing else can be supposed on the polynomial because if $\FFp=\FF_3$ for example, any polynomial of $\FFp[\lambda]$ can appear.
This also shows that $\lambda$ cannot be found in the ground field. Suppose an extension of degree $(s-n)$ is given as $\FFp[X]/(P)$ where $P$ is an
irreducible degree-$(s-n)$ polynomial with coefficients in $\FFp$. Then a root of $P$ in $\FFp[X]/(P)$ cannot be a root of a degree-$(s-n-1)$
polynomial with coefficients in $\FFp$. Thus, if one can find such a $P$, taking for $\lambda$ the indeterminate $X$ is sufficient. For any fixed
characteristic $p$, Shoup gives a deterministic polynomial-time algorithm \cite{Sho90} that given an integer $N$ outputs a degree-$N$ irreducible 
polynomial $P$ in $\FFp[X]$. Thus, the system $g$ is now a square system of polynomials in $\left(\FFp[X]/(P)\right)[\bar x,\bar y]$ and this system
has a non trivial root iff $f$ has a non trivial root. And Shoup's algorithm allows us to build $g$ in polynomial time from $f$.

For any field \FFp, it has been shown that from an instance $\mathcal B$ of \lang{Boolsys} a square system $g$
of polynomials with coefficients in an extension of \FFp 
(in \ZZ for integer polynomials)
can be built in deterministic polynomial time such that $g$ has a non trivial root iff 
$\mathcal B$ is satisfiable. This shows that the problem is $\NP$-hard.  
%
%
\qed\end{proof} 

The previous result is somewhat unsatisfactory as it requires, in the case of positive characteristic, to work with coefficients in an extension field rather than in the
ground field. A way to get rid of this limitation is now shown. Yet, a property of the previous result is lost. Instead of having constant-degree (even
degree-$2$) polynomials, our next result uses linear-degree polynomials. 
It is not clear whether the same result can be obtained 
for degree-$2$ polynomials (for instance, as explained at the end of Sect.~\ref{char0} reducing the degree by introducing new variables can create unwanted
solutions at infinity).

The basic idea behind Theorem~\ref{groundfield} is quite simple (we put the irreducible polynomial used to build the extension field into the system), but some care is required in order to obtain
an equivalent homogeneous system.

Recall from the introduction that $\hhnsq(\FF_p)$ is the following problem:
\begin{itemize}
\item[-] INPUT: a square system of homogeneous equations with coefficients
in $\FF_p$.
\item[-] QUESTION: is the system satisfiable in the algebraic closure of $\FF_p$?
\end{itemize}
\begin{theorem} \label{groundfield}
For any prime $p$, $\hhnsq(\FF_p)$ is $\NP$-hard 
under deterministic reductions.
\end{theorem}

\begin{proof} 
The idea for this result is to turn coefficient $\lambda$ in the previous proof into a variable and to add the polynomial $P$ as a component of the
system. Of course, considering $\lambda$ as a variable implies that the polynomials are not homogeneous anymore. Thus, it remains to explain how to keep
the system homogeneous.

First, the polynomial $P$ needs to be homogenized. This is done through the variable $x_0$ in the canonical way. As $P(\lambda)$ is irreducible, it is in 
particular not divisible by $\lambda$. Hence, the homogenized polynomial $P(\lambda,x_0)$ contains a monomial $\alpha\lambda^d$ and another one 
$\beta x_0^d$ where $d$ is the degree of $P$. Hence $x_0$ is zero iff $\lambda$ is.

The other polynomials have the form $f_{n+i}(\bar x)-y_{i-1}^2+\lambda y_i^2$. It is impossible to homogenize those polynomials by multiplying $f_{n+i}$
and $y_{i-1}^2$ by $x_0$ (or any other variable) because then the variable $y_{i-1}$ never appears alone in a monomial, and a $s$-tuple with all variables
set to $0$ but $y_{i-1}$ would be a non trivial solution. Moreover, in the previous proof, the fact that the $y_i$ all appear with degree $2$ is used to
consider the system as a linear system in the $y_i^2$. Thus replacing the monomial $\lambda y_i^2$ by $\lambda y_i$ does not work either. 
Instead, we construct the slightly more complicated homogeneous system:
\begin{equation} g_h(\bar x,\bar y,\lambda)=\left(\begin{array}{l@{}c@{}l}
f_1(\bar x)&&\\
\quad\vdots&&\\
f_n(\bar x)&&\\
x_0^{s-n-1}f_{n+1}(\bar x)&&+\lambda y_1^{s-n}\\
x_0^{s-n-2}f_{n+2}(\bar x)&-y_1^{s-n}&+\lambda y_2^{s-n-1}\\
&\vdots&\\
x_0^{s-n-i}f_{n+i}(\bar x)&-y_{i-1}^{s-n-i+2}&+\lambda y_i^{s-n-i+1}\\
&\vdots&\\
x_0f_{s-1}(\bar x)&-y_{s-n-2}^3&+\lambda y_{s-n-1}^2\\
f_s(\bar x)&-y^2_{s-n-1}&\\
P(\lambda,x_0)
\end{array}\right) \end{equation}
Contrary to the previous proof, the $y_i$ do not appear all at the same power. Yet, all the occurrences of each $y_i$ 
have the same degree, and we shall prove that this is sufficient.

Let us prove that if $f$ does not have any non trivial root, then neither does $g_h$. Some of the observations made
for $g$ in the previous proof remain valid. Hence, it is sufficient to prove that a non trivial $(s+1)$-tuple $(\bar a,\bar b,\ell)$ cannot be solution of $g_h$ whenever
$a_0=1$, $\bar b\neq\bar 0$ and $a_0^2=\cdots =a_n^2$ if $p\neq 2$ or $a_i\in \{0,a_0\}$ if $p=2$.
By a previous remark on the polynomial $P$, $\ell$ can also be supposed to be nonzero.

So, similarly as in the previous proof, let us define $\epsilon_i=a_0^{s-n-i}f_{n+i}(\bar a)\in\FFp$. 
In the system $g_h$, the variable $y_i$ only appears at the power 
$(s-n-i+1)$. Therefore, given a value of $\bar a$ and $\ell$, the tuple $(\bar a,\bar b,\ell)$ is a root of $g_h$ iff the $b_i^{s-n-i+1}$ satisfy the 
linear system
\begin{equation} \left\{\begin{matrix}
\epsilon_1 && &+&  \ell Y_1 &=&0\\
\epsilon_2 &-& Y_1 &+& \ell Y_2 &=&0\\
&&&\vdots\\
\epsilon_{s-n-1} &-& Y_{s-n-2} &+& \ell Y_{s-n-1} &=&0\\
\epsilon_{s-n} &-& Y_{s-n-1} && &=&0
\end{matrix}\right. \end{equation}
This is the same system as in the previous proof. Now if $(\ell,1)$ is supposed to be a root of $P$, as $P$ is an irreducible polynomial of degree $(s-n)$,
$\ell$ cannot be a root of a univariate polynomial of degree less than $(s-n)$ with coefficient in $\FFp$. But the determinant of the linear system is such
a polynomial, and thus cannot be zero. This determinant is then $0$ iff all the $\epsilon_i=0$. The same arguments as in the previous proof can be used to
conclude that $(\bar a,\bar b,\ell)$ can be a root of $g_h$ iff $\bar a$ is a root of $f$.

Thus, from an instance $\mathcal B$ of \lang{Boolsys}, a square homogeneous system $g_h$ of polynomials with coefficients in the ground field \FFp is built
in deterministic polynomial time. This system has a non trivial root iff $\mathcal B$ is satisfiable. The result is proved.
%
%
%
\qed\end{proof} 



\section{Final remarks}

In characteristic zero,
the upper and lower bounds on \hhnsq are in a sense close to each other. Indeed, 
$\NP\subseteq\AM\subseteq\mathsf{\Pi}_2\P$, 
that is, $\AM$ lies between the first and the second level of the polynomial hierarchy. 
Furthermore, ``under plausible complexity conjectures, $\AM=\NP$'' \cite[p157]{AroraBarak}.
Improving the $\NP$ lower bound may be challenging as the proof of 
Proposition \ref{prop:in-am} shows that this would imply the same lower bound for  \emph{Hilbert's Nullstellensatz}.

In positive characteristic, the situation is quite different. Indeed, the best known upper bound for \emph{Hilbert's Nullstellensatz} as well as for
the resultant is \PSPACE. 
As in characteristic zero, the known upper and lower bounds are therefore 
the same for 
both problems. 
But as the gap between the $\NP$ lower bound and the $\PSPACE$ upper bound
is rather big, these problems might be of widely different complexity
(more precisely, testing the resultant for zero could in principle be much easier than
deciding whether a general polynomial system is satisfiable).
Canny's algorithm for computing the resultant \cite{Can88-robot} involves the
computation of 
the determinants of 
exponential-size matrices, known as Macaulay matrices, in polynomial space. Those matrices admit a succinct representation
(in the sense of \cite{GW84}). One can prove that computing the determinant
of a general succinctly represented matrix is $\mathsf{FPSPACE}$-complete (and even testing for zero is $\PSPACE$-complete) \cite{Gre09}.
It follows that the $\mathsf{FPSPACE}$ upper bound could be improved only by exploiting
the specific structure of the Macaulay matrices in an essential way, or
by finding an altogether different (non Macaulay-based) 
approach to this problem. As pointed out in Section~\ref{char0}, in characteristic zero  a different approach is indeed possible for {\em testing whether the resultant vanishes} (rather than for computing it).
This problem is wide open in positive characteristic.

Finally, an interesting open question is whether the randomized reduction of Theorem~\ref{thm:randomized} can be derandomized.

\subsubsection{Acknowledgments.}
We thank Bernard Mourrain and Maurice Rojas for sharing their insights
on the complexity of the resultant in characteristic 0.

\appendix
\section{Appendix}

The following result is used in the proof of Theorem~\ref{thm:randomized}.
\begin{theorem}\label{thm:variety}
Let \KK be an algebraically closed field and $V$ an algebraic variety of $\KK^{n+1}$ defined by a set of homogeneous degree-$d$ polynomials 
$f_1,\dots,f_s\in\KK[x_0,\dots,x_n]$. This variety can be defined by $(n+1)$ homogeneous degree-$d$ polynomials $g_1,...,g_{n+1} \in \KK[x_0,\dots,x_n]$.
Moreover, suitable $g_i$'s  can be obtained by taking generic linear combinations of the $f_i$.  That is, we can take $g_i=\sum_{j=1}^s \alpha_{ij}f_j$
where $(\alpha_{ij})$ is a matrix of elements of $\KK$, 
and the set of suitable matrices is Zariski-dense in $\KK^{s(n+1)}$.
\end{theorem}

\begin{proof}
A similar result is established in  \cite[Proposition 1]{Koi00-circuits} for arbitrary (possibly non-homogeneous) polynomials: Any algebraic variety of 
$\KK^{n}$ can be defined by taking $n+1$ generic linear combinations of the original equations.
In Theorem~\ref{thm:variety} the  polynomials are assumed to be of same degree to ensure that the linear combinations $g_i$ are homogeneous. If the $f_i$ are not of the same degree, the system can be transformed into an equivalent system where 
the degree of all polynomials is equal to the least common multiple  of the degrees of the polynomials in the original system (in the application to 
Theorem~\ref{thm:randomized}, this transformation is not necessary since the input system is made of polynomials of equal degree).

Let  $V_\alpha$ be the variety defined by the $g_i$. 
Clearly, $V\subseteq V_\alpha$ for any matrix $\alpha$.
We need to show that there is a Zariski-dense set of matrices $\alpha$
such that $V_\alpha \subseteq V$. 

Consider a point $(a_0,\dots,a_n)\in \KK^{n+1}$.
If $a_0=0$, then one can define new polynomials $\tilde f_j$ and $\tilde g_i$ 
by setting $\tilde f_j(x_1,\dots,x_n)=f_j(0,x_1,\dots,x_n)$ and 
$\tilde g_i(x_1,\dots,x_n)=g_i(0,x_1,\dots,x_n)$.
The new polynomials satisfy the same linear relations, namely, 
we have $\tilde{g_i}=\sum_{j=1}^s \alpha_{ij}\tilde{f_j}$. 

By \cite[Proposition 1]{Koi00-circuits}, there is a Zariski-dense set of matrices $\alpha$ such that for any tuple  $(a_1,\dots,a_n)$,
if the $\tilde g_i$ vanish on $(a_1,\dots,a_n)$ the same is true of the $\tilde f_j$. 
In this case the
$f_j$ vanish on $(0,a_1,\dots,a_n)$, and therefore $(a_0,\dots,a_n)\in V$. 

It remains to examine the case $a_0\neq 0$.
In this case, for any tuple $(a_0,a_1,\dots,a_n) \in V_{\alpha}$, 
we have $(1,a_1/a_0,\dots,a_n/a_0)\in V_\alpha$ since the polynomials are homogeneous. The same argument as for the case $a_0=0$ shows that
we will have $(1,a_1/a_0,\dots,a_n/a_0)\in V$ for  $\alpha$ in a Zariski-dense set (we now apply  \cite[Proposition 1]{Koi00-circuits} to the polynomials
 $\tilde f_j(x_1,\dots,x_n)=f_j(1,x_1,\dots,x_n)$ and 
$\tilde g_i(x_1,\dots,x_n)=g_i(1,x_1,\dots,x_n)$). 
By appealing again to homogeneity, we can conclude that $(a_0,\dots,a_n)\in V$.

From the above analysis it follows that any matrix $\alpha$ belonging
to the intersection of two Zariski-dense sets (corresponding to the two cases
$a_0=0$ and $a_0 \neq 0$) is suitable. This concludes the proof since a finite intersection of Zariski-dense sets is Zariski-dense.\qed
\end{proof}

\end{document}